\documentclass[11pt,letterpaper]{article}

\title{\bf Holographic Special Relativity}
\author{
{\bf Derek K.\ \!Wise} \\[.5em]
{\sl \small Institute for Quantum Gravity} \\[-.3em]
{\sl \small Universit\"at Erlangen--N\"urnberg} \\[-.3em]
{\sl \small Staudtstr.\ \!7/B2,\! 91058 Erlangen,\! Germany} \\
\small \texttt{derek.wise@gravity.fau.de} 
}
\date{May 14, 2013}

\usepackage{amssymb,amsmath}
\usepackage[
    colorlinks,%
    linkcolor=blue,citecolor=red,urlcolor=blue,
]{hyperref}
\usepackage[all,knot]{xy}
\usepackage{tikz}
\usetikzlibrary{arrows,backgrounds}
\tikzset{small/.style={font=\fontsize{9}{9}\selectfont}}
\tikzset{vsmall/.style={font=\fontsize{7}{7}\selectfont}}
\usepackage{graphicx}

\renewcommand{\S}{\mathcal{S}} 

\newcommand{\ff}{\mathcal{F}} 
\newcommand{\fo}{\mathcal{O}} 

\newcommand{\Met}{\mathcal{G}} 

\newcommand{\iO}{\overline{\mathcal{O}}} 

\newcommand{\GH}{S^{3,1}} 
\newcommand{\GK}{O} 
\newcommand{\GKt}{\overline{O}}
\newcommand{\HK}{\mathrm{H}^3} 

\newcommand{\cone}{\mathcal{C}}
\newcommand{\pastcone}{\mathcal{C}^{-}}
\newcommand{\futurecone}{\mathcal{C}^{+}}

\newcommand{\GP}{P(\cone)} 

\newcommand{\ggh}{\mathfrak{z}} 
\newcommand{\ghk}{\mathfrak{y}} 

\newcommand{\gm}{\mathfrak{g}_{\scriptscriptstyle{-1}}} %
\newcommand{\Go}{G_{\scriptscriptstyle{0}}} %
\newcommand{\go}{\mathfrak{g}_{\scriptscriptstyle{0}}} %
\newcommand{\gp}{\mathfrak{g}_{\scriptscriptstyle{+1}}} %
\newcommand{\gpm}{\mathfrak{g}_{\scriptscriptstyle{\pm}}} %

\newcommand{\Hyp}{\mathrm{H}^4} 

\newcommand{\past}{a_{-}} 
\newcommand{\future}{a_{+}} 

\newcommand{\Ad}{{\rm Ad}}

\newcommand{\Hpp}{H''}
\newcommand{\hpp}{\mathfrak{h}''}

\newcommand{\Kt}{\Go}

\newcommand{\SIM}{{\rm SIM}}

\newcommand{\GL}{{\rm GL}}

\newcommand{\SO}{{\rm SO}}
\newcommand{\so}{\mathfrak{so}}

\newcommand{\ISO}{{\rm ISO}}

\newcommand{\g}{\mathfrak{g}}
\newcommand{\h}{\mathfrak{h}}
\newcommand{\p}{\mathfrak{p}}

\renewcommand{\k}{\mathfrak{k}}

\newcommand{\R}{\mathbb{R}}

\newcommand{\Z}{\mathbb{Z}}

\newcommand{\maps}{\colon}

\newcommand{\define}[1]{{\bf #1}}

\newtheorem{thm}{Theorem}

\newtheorem{prop}[thm]{Proposition}
\newtheorem{defn}[thm]{Definition}

\newenvironment{proof}{\noindent
\textbf{Proof: }}{\hfill\rule{.6em}{.8em} \medskip}
\newenvironment{proof.within.proof}
{\noindent{\it Proof:}}{
\hfill $\Box$ \medskip}

\newcounter{Ccounter} 
\newenvironment{C-list}{  
\begin{list}{{\rm C\arabic{Ccounter}}.}{\usecounter{Ccounter}}
}{\end{list}}

\newcommand{\arxiv}[1]{\href{http://arxiv.org/abs/#1}{arXiv:\nolinkurl{#1}}}
\newcommand{\narxiv}[2]{\href{http://arxiv.org/abs/#1}{arXiv:\nolinkurl{#1 [#2]}}}

\begin{document}

\maketitle

\thispagestyle{empty}

\begin{abstract}\noindent
We reinterpret special relativity, or more precisely its de Sitter deformation, in terms of 3d conformal geometry, as opposed to (3+1)d spacetime geometry.  An inertial observer, usually described by a geodesic in spacetime, becomes instead a choice of ways to reverse the conformal compactification of a Euclidean vector space up to scale.  The observer's `current time,' usually given by a point along the geodesic, corresponds to the choice of scale in the decompactification.  We also show how arbitrary conformal 3-geometries give rise to `observer space geometries,' as defined in recent work, from which spacetime can be reconstructed under certain integrability conditions.  We conjecture a relationship between this kind of `holographic relativity' and the `shape dynamics' proposal of Barbour and collaborators, in which conformal space takes the place of spacetime in general relativity.   We also briefly survey related  pictures of observer space, including the AdS analog and a representation related to twistor theory.  
\end{abstract}  

\section{Introduction}

Minkowski introduced the idea of {\em spacetime} in 1908 as a conceptual framework for Einstein's theory of special relativity, published three years before \cite{minkowski}.  As radical as it then seemed, the idea of spacetime is today hardly questioned, and indeed plays essential roles in the two current pillars of fundamental physics: quantum field theory and general relativity.  In fact, unlike special relativity, the very idea of spacetime was fundamental in the invention of general relativity. 

Yet despite the indelible impression Minkowski's insight has left on our thinking, relativity is not fundamentally about spacetime.  It is foremost concerned with how different observers view the world around them, and how their views relate to each other.  The spacetime perspective is compelling precisely because it efficiently accounts for all these possible viewpoints.  Namely, starting with spacetime, the \define{observer space}---the space of all possible observers in the universe---is simply the space of future-directed unit timelike vectors.   But spacetime is not the only framework for understanding observers.  

In this paper we present an alternative way of encoding observers using not Lorentzian but rather  {\em conformal} geometry.  As in (A)dS/CFT, the key to this relationship is a certain Lie group coincidence: the connected de Sitter group $$G:=\SO_o(4,1)$$ is not only the proper orthochronous isometry group of de Sitter spacetime $\GH$ but also the group of symmetries of the conformal 3-sphere, represented as a projective light cone  $\GP$ in $(4+1)$-dimensional Minkowski spacetime.   The relationship between these two spaces is subtle.  While both are homogeneous $G$-spaces, and $\GP$ can be thought of as the past or future boundary of $\GH$, there is no $G$-equivariant map between $\GH$ and $\GP$, and hence no direct way of mapping things happening in spacetime to things happening on the boundary, or vice versa, in a way that respects the symmetries of both.

On the other hand, observers take priority over spacetime.  The observer space of de Sitter spacetime is also a homogeneous $G$-space and,  while there is no equivariant map between $\GH$ and $\GP$, there is the next best thing---a span of equivariant maps, with observer space $\GK$ at the apex:
\begin{equation*}
\label{span}
  \begin{tikzpicture}[scale=1,>=stealth']
      \node (spacetime) at (0,0) {$\GH$};
      \node (conf) at (2,0) {$\GP$};
      \node (obs) at (1,1.4) {$\GK$}
          edge [->]  node [above] {$\pi$} (spacetime)
          edge [->]  node [above] {$a$} (conf);
  \end{tikzpicture}
\end{equation*}
The map on the left is just a restriction of the tangent bundle: it sends each observer to its base event in spacetime.  

The most direct way to describe the map on the right is to say it sends each observer to the asymptotic past of its geodesic extension.  This is true, but not quite satisfactory since it relies on spacetime, and hence on the map $\pi\maps \GK \to \GH$ for its definition.  To explore the idea of encoding observers using conformal geometry rather than spacetime geometry, we would prefer a description of $a\maps \GK \to \GP$ that does not involve spacetime as an intermediate step.  

To arrive at a more intrinsic description, note that while de Sitter spacetime has a conformal 3-sphere as its asymptotic boundary, this boundary seems to any given {\em observer} to have more than a conformal structure: it appears to be an ordinary sphere of infinite radius, effectively  a three-dimensional affine space, once we delete the point antipodal to the observer's asymptotic past.   Using the observer's asymptotic past as the `origin,' this affine space becomes a vector space.  

This `vector space in the infinite past' has an inner product only up to scale, since the magnitude  diverges as we push further back into the past.  However, we can renormalize the scale of this vector space by declaring the unit sphere to be at the asymptotic past of the observer's current light cone.  This makes the boundary, after removing the `point at infinity,' an inner product space, or Euclidean vector space. 

In summary, an observer comes with a particular idea of how to put a Euclidean vector space structure on the complement of a point in the conformal sphere.  For short, we can refer to this process as \define{Euclidean decompactification}, since it reverses the conformal compactification of a Euclidean vector space.  Remarkably, observers in de Sitter spacetime correspond one-to-one with such Euclidean decompactifications of the sphere.   Ultimately, we get two isomorphic perspectives on the de Sitter observer space:
\begin{enumerate}
\item Observers live in de Sitter spacetime.  Different observers are distinguished by their velocity vectors, so observer space is the space of all unit future-directed timelike vectors.  Conformal space is an auxiliary construction given by identifying all observers whose geodesic extensions are asymptotic in the past.  

\item Observers live in the conformal 3-sphere.   Different observers are distinguished by their preferred  Euclidean decompactification, so observer space is the space of all such decompactifications.  Spacetime is an auxiliary construction given by identifying all observers who share the same co-oriented unit sphere.  
\end{enumerate}
A co-orientation of a 2-sphere embedded in $\GP$ is just a choice which complementary component is `inside' the 2-sphere, here determined by the `origin' of a given observer.

Part of this article is concerned with fleshing out the above ideas in more detail, and is essentially expository in nature.  Many authors have nicely explained the geometry of de Sitter spacetime (see e.g.\ \cite{Coxeter,HawkingEllis,sightseeing} for an interesting sample), and there is necessarily some overlap here.  However, the purpose is not to repeat available explanations but rather to explain the geometry of {\em observers} in the de Sitter universe, and in particular the relation between observers and conformal geometry.  I am not aware of any similar exposition.   

A more serious goal is to provide new examples of `observer space geometries.'  In previous work \cite{lifting}, we have studied deformations of the de Sitter observer space, using Cartan geometry to provide a general definition of `observer space' flexible enough to unify the geometric treatment of a wide variety of theories of space and time.  Any solution of general relativity is an example, but so are the observer spaces of spacetime theories with preferred foliations \cite{gielen}, galilean spacetime, Finsler spacetime \cite{hohmann}, and models with no invariant notion of spacetime at all \cite{lifting}.  Here we will see that observer space geometries also arise `holographically' from 3-dimensional conformal geometry.  The main mathematical result is Theorem \ref{mainthm}, giving a canonical Cartan geometry on a certain bundle over any 3d conformal manifold, modeled on the observer space of de Sitter spacetime.

\section{Observers in de Sitter spacetime}

The \define{observer space} of a time-oriented Lorentzian spacetime is the space of normalized future-directed timelike tangent vectors \cite{lifting}.  For the moment, we are mainly concerned with $(3+1)$-dimensional {de Sitter spacetime} and its 7-dimensional observer space.  Both of these spaces can be conveniently described using $\R^{4,1}$, the vector space $\R^5$ equipped with the standard Minkowski inner product $\eta$ of `mostly positive' signature.   \define{De Sitter spacetime} is the pseudosphere $\GH\subset \R^{4,1}$ of all points with spacelike distance $\sqrt{3/\Lambda}$ from the origin; here we normalize the cosmological constant $\Lambda$ so that this spacelike radius is 1.   Thus:
\[
   \GH = \{x \in \R^{4,1} : \eta(x,x) = 1\}
\]
For the empirical value of $\Lambda$ in our universe, this means using units of about $2\times 10^{28}$ meters, or  equivalently $6.7\times 10^{19}$ seconds, since we also set the speed of light to 1.  These numbers indicate just how minuscule the deviation from Minkowski spacetime in common units.

We will make much use of linear and affine subspaces of $\R^{4,1}$.  If $v,w,\ldots$ are nonzero vectors in $\R^{4,1}$, we denote their linear span by $[v,w,\ldots]$.   The light cone at a point $x$ in de Sitter spacetime is the intersection of de Sitter spacetime with the affine plane $x+[x]^\perp$:
\begin{equation}
\label{lightconeplane}
\begin{tikzpicture}[small,baseline=(current bounding box.center)]
    \node[anchor=south west,inner sep=0] (image) at (0,0) {\includegraphics[height=3.7cm]{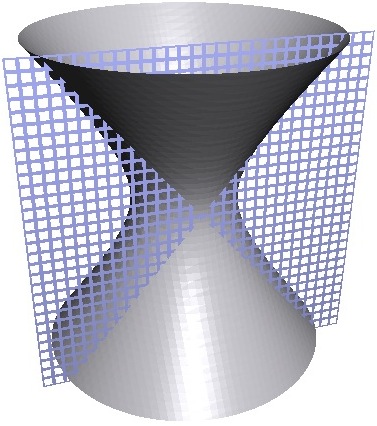}};
    \begin{scope}[x={(image.south east)},y={(image.north west)}]
    \node [blue] (event) at  (0.54,0.485) {$\bullet$};
    \node [blue] at (.54,.4) {$x$};
    \node (plane) at (1.2,0.8) {$x+[x]^\perp$};
    \draw [->,>=stealth] (1.5,.4)--(1.5,.6); \node at (1.6,.5)  {time};
    \end{scope}
\end{tikzpicture}
\end{equation}

This same picture helps describe the observer space of de Sitter spacetime.  At a point $x\in \GH$, the tangent space $T_x \GH$ may be identified with $x+[x]^\perp$. 
An observer is a unit future-timelike vector, so translating the tangent plane $x+[x]^\perp$ back to the origin of $\R^{4,1}$, we can think of particle velocities at $x$ as living in $[x]^\perp\cap \Hyp$, where 
\[
   \Hyp = \{u \in \R^{4,1} : \eta(u,u) = -1,\, u_0 > 0\}
\]
is 4-dimensional hyperbolic space.  This lets us see \define{de Sitter observer space} as a subspace of $\GH\times\Hyp$:
\[
   \GK = \{(x,u) \in \GH\times \Hyp : \eta(x,u) = 0\}.
\]
The projection $(x,u)\mapsto x$ is a fiber bundle over spacetime, each fiber isomorphic to $\HK$. 

The \define{surface of simultaneity} of the observer $(x,u)$ is defined, just as in Minkowski spacetime, to be the totally geodesic 3-dimensional surface orthogonal to the observer's velocity; this is just $[u]^\perp\cap M$.  Unlike their Minkowski analogs, the surfaces of simultaneity of a given inertial observer at different times are not disjoint.  This leads to some peculiar large-scale behavior: two observers in distant parts of the universe can view the same `surfaces of simultaneity' as occurring in the opposite chronological order.   Coxeter's article \cite{Coxeter} begins with a Lewis Carroll quote, presumably for this reason.  Like the `paradoxes' of Minkowskian special relativity, this causes no problems: two such observers are on opposite sides of each other's cosmological horizon.  In any case, the space of these surfaces of simultaneity, the space of possible `nows' is just $\Hyp$.

What we have termed an observer might be more accurately called an `instantaneous observer.'   In contrast, an \define{inertial observer} is a timelike geodesic, and the space of these is the \define{inertial observer space}, $\GKt$.  The geodesic through the observer $(x,u)$ is one of the two curves at the intersection of $\GH$ with the subspace $[x,u]$:
\begin{center}
\includegraphics[height=5cm]{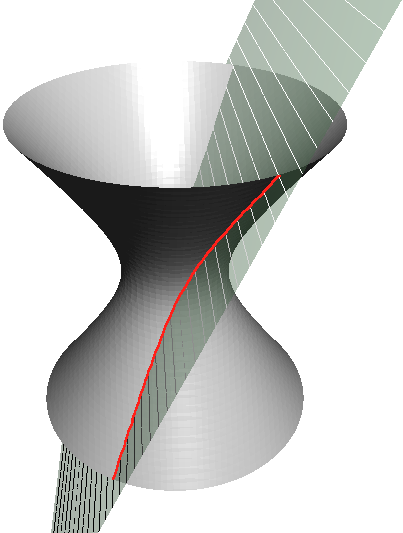}
\end{center}
The proper-time parameterization of this geodesic, starting at $x\in \GH$, is 
\begin{equation}
\label{geodesic}
    \gamma(\tau) = (\cosh \tau) x + (\sinh \tau) u,
\end{equation}
and the path followed by the inertial observer through observer space $\GK$ is $\tau \mapsto (\gamma(\tau),\dot\gamma(\tau))$, where the dot denotes the $\tau$ derivative.

We now turn to 3-dimensional conformal space, of which a well known model is the \define{projective light cone} $\GP$: the space of one-dimensional null subspaces of $\R^{4,1}$.   Let us recall how $\GP$ gets its conformal structure.  If $\pastcone\subset \R^{4,1}$ is the past light cone, we have the canonical map 
\begin{equation}
\label{pastlightcone}
  \begin{tikzpicture}[>=stealth,baseline=(current bounding box.center)]
     \node (base) at (0,0) {$\GP$};
     \node  at (0,1) {$\phantom{{}^{-}}\pastcone$}
        edge [->] (base);
     \node [small] (base-elt) at (1.3,0) {$[v]$};
     \node [small] at (1.3,1) {$v$}
        edge [|->] (base-elt);
  \end{tikzpicture}
\end{equation}
The tangent space of $\pastcone$ at $v$ is canonically isomorphic to $[v]^\perp$, via translation along $v$, and all vectors in $[v]^\perp$ that are not proportional to $v$ are spacelike.  Thus, if $q$ is any section of (\ref{pastlightcone}), the pullback of the metric along that section is a Riemannian metric.  Moreover, if we adjust the section $q$ by multiplying it by a positive function on $\GP$, the corresponding metric just changes by the same multiple.   Thus $\pastcone$ is isomorphic to the tautological bundle of a conformal metric on $\GP$ \cite{FG}.

Two canonical maps from de Sitter observer space to the conformal 3-sphere are given by sending a given observer inertially into the infinite past or future.  These maps can be nicely visualized using the ambient space $\R^{4,1}$.  First, an observer gives a subspace $[x,u]$ of $\R^{4,1}$, intersecting $\GH$ along the corresponding inertial observer's worldline (\ref{geodesic}) and the antipodal worldline.  Then, forgetting about $\GH$,  this same subspace $[x,u]$ intersects the light cone in a pair of null lines; one of these is  asymptotic in the past to the worldline, while the other is asymptotic in the future.  We can draw the past asymptotic map as
\begin{equation}
\label{pastmap}
\begin{tikzpicture}[baseline=(current bounding box.center)]
    \node[anchor=south west] (image) at (0,0) {\includegraphics[height=5cm]{dsplane2.png}};
    \node[anchor=south west] (image2) at (5,0) {\includegraphics[height=5cm]{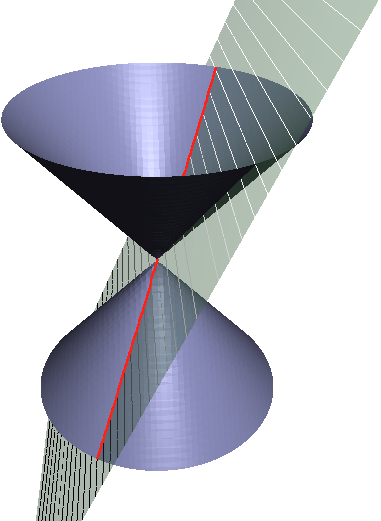}};
   \node (hyperbola) at (1.45,1.2) {};
   \node (lightray) [above of=image2] {};
   \node (geo-label) [small,text width=4cm, text badly ragged] at (-1.5,3) {\sf timelike geodesic determined by the observer};
   \node [below of=geo-label, node distance=2em,small] {$\gamma(\tau) = (\cosh \tau) x + (\sinh \tau) u$}
          edge [-, very thin] (hyperbola);          
   \node (confpt-label) [small] at (10,3) {\sf null subspace}
          edge [-, very thin] (lightray);
   \node [below of=confpt-label, node distance=1em,small] {$[x-u]\in \GP$};     
   \node (plane1) at (2.3,2.8) {};
   \node (plane2) at (6.5,2.6) {};  
   \node [small] at (4.5,2) {$[x,u]$}
      edge (plane1)
      edge (plane2);
\end{tikzpicture}
\end{equation}
The future asymptotic map is the same except that the other light ray in $\R^{4,1}\cap [x,u]$ is chosen.  We see that the past and future asymptotic maps are given by  
\[
 \past,\future\maps \GK \to \GP,\qquad 
 \begin{array}{c@{\,=\,}l}
     \past(x,u) & [x-u]\\
     \future(x,u) & [x+u].
 \end{array}
\]
Note that the past and future boundaries are naturally identified, so that both of these maps land in the same copy of conformal space.  

The map $\past\maps \GK \to \GP$ is clearly many-to-one.  In fact, we will see that the preimage of an point in conformal space determines a \define{local field of observers}, a local section of the canonical projection $\GK \to \GH$ from observer space to spacetime.  Using the description of $\GK$ as a subset of $\GH\times \Hyp$, we can think of an observer field as a map $u$ from some region of $\GH$  to $\Hyp$, such that $\eta(x,u(x))=0$ for all $x$.

To find these local observer fields explicitly, pick $[v]\in\GP$, where $v\in \pastcone$.  If an observer $(x,u)$ satisfies $\past(x,u)=[v]$, then since $x-u\in \pastcone$, we have $v = e^t(x-u)$ for some $t\in\R$.  Taking the inner product with $x$, we find that 
\begin{equation}
\label{FRWtime}
e^t={\eta(x,v)},
\end{equation}
and hence in particular that $\eta(x,v)>0$.  This determines a region of de Sitter spacetime above the  subspace $[v]^\perp$: 
\begin{center}
\begin{tikzpicture}
    \node[anchor=south west,inner sep=0] (image) at (0,0) {\includegraphics[width=4cm]{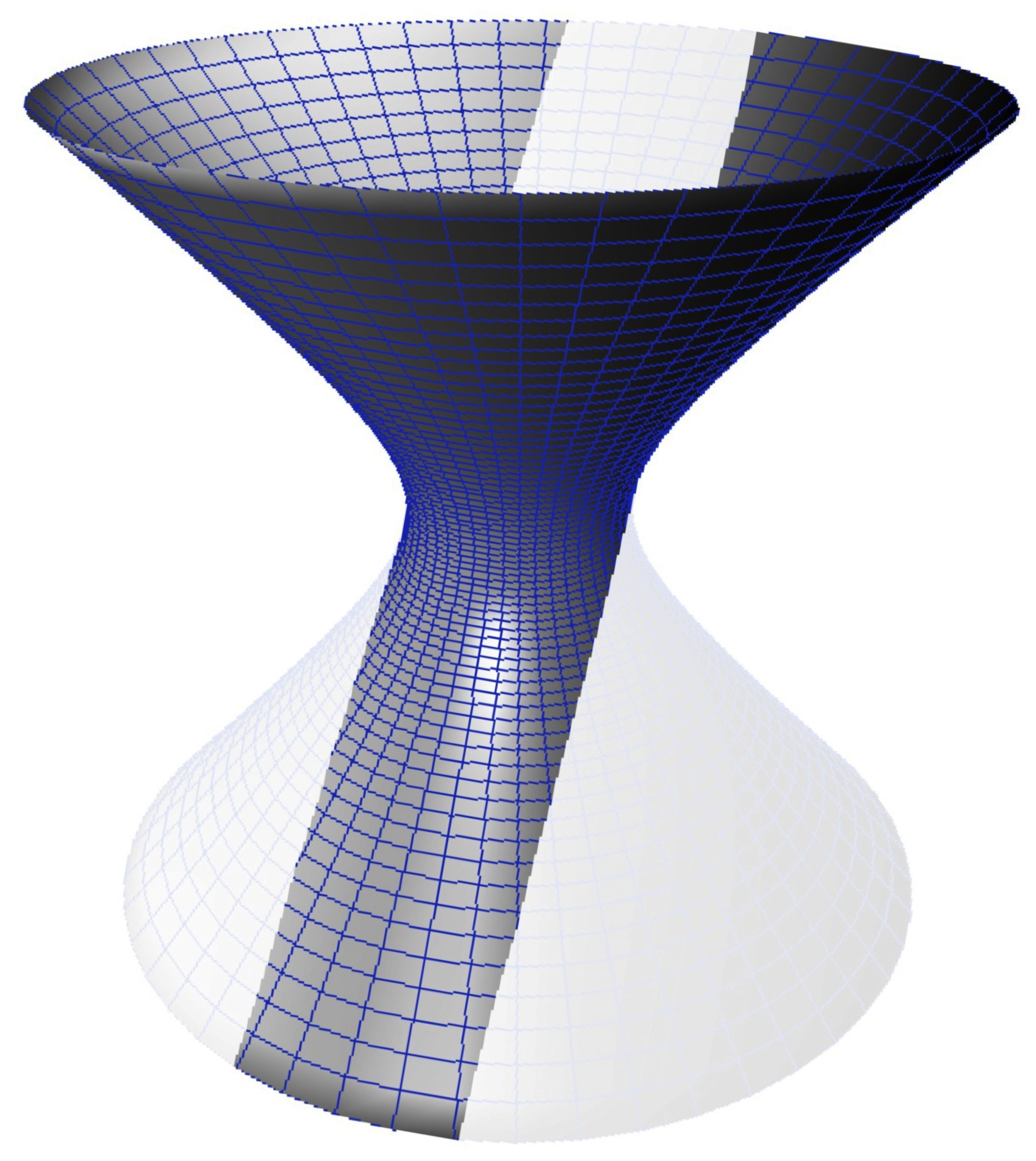}};
    \begin{scope}[x={(image.south east)},y={(image.north west)}]
     \node (patch) at (1.15,0.3) {$\{x\in \GH : \eta(x,v)>0\}$}
        edge (.5,.4);
    \end{scope}
\end{tikzpicture}
\end{center}
On this region, we can solve uniquely for $u$, obtaining a local field of observers given by: 
\begin{equation}
\label{observer-field}
    u(x) = x-\frac{v}{\eta(x,v)}
\qquad \forall x \in \GH  \text{ with } \eta(x,v)>0.
\end{equation}

The observers in this local observer field all share the same \define{past horizon}, the boundary of the future of the intertial observer's worldline.  Using the geodesic (\ref{geodesic}) followed by an inertial observer, the light cone at time $\tau$ (see (\ref{lightconeplane})) is $(\gamma(\tau) + [\gamma(\tau)]^\perp)\cap \GH$, and this tends to $[x-u]^\perp\cap \GH$ as $\tau\to -\infty$.  Since the observer field (\ref{observer-field}) consists of all observers with the same asymptotic past $[x-u(x)]$, all of these observers agree on the past horizon.  

Level surfaces of the  time coordinate $t$ are intersections of de Sitter spacetime with hyperplanes in $\R^{4,1}$ parallel to $[v]^\perp$; they have no intrinsic curvature.   Two observers have the same time coordinate $t$ if and only if they map to the same point in $\pastcone$ under the map $(x,u)\mapsto x-u$. 

It is worth noting that the observer field (\ref{observer-field}) is related to a well known local coordinate system defined on the same region.  For this, let us specialize to the case where $v$ is the point on the light cone with coordinates $(-1,0,0,0,1)$.  On the region $\eta(x,v)>0$, i.e.\ $x^0 + x^4 >0$, we define coordinates
\[
 t= \log(x^0 + x^4) \qquad
 y^i =\frac{x^i}{x^0 + x^4},\quad i=1,2,3
\]
where the time coordinate $t$ was already introduced in (\ref{FRWtime}).  
It is then easy to check that the vector field 
\begin{equation}
\label{vf}
   \frac{\partial}{\partial t}
=
    \left(x^0 + \frac{1}{x^0 + x^4}\right)\frac{\partial}{\partial x^0} +
    x^i \frac{\partial}{\partial x^i} +
    \left(x^4 -  \frac{1}{x^0 + x^4}\right)\frac{\partial}{\partial x^4}
\end{equation}
in these coordinates coincides with the observer field (\ref{observer-field}) determined by our choice of $v$. Coordinate systems for different choices of null subspace $[v]\in \GP$ can be obtained by Lorentz transformations of $\R^{4,1}$. 
However, for us, any single one of these coordinate systems is not important; we care only that the collection of local vector fields (\ref{observer-field}) is a conformal 3-sphere.

\section{Symmetries}
We have discussed relationships among de Sitter observer space, spacetime, and conformal space, but have so far ignored their symmetries.  Let $G$ denote the connected de Sitter group 
$$G=\SO_o(4,1).$$ 
All of the spaces we have discussed---de Sitter spacetime $\GH$, its observer space $\GK$ and inertial observer space $\GKt$, 4-dimensional hyperbolic space $\Hyp$, the past cone $\pastcone\subset\R^{4,1}$ and the conformal 3-sphere $\GP$---are homogeneous $G$-spaces.  Moreover, it is clear that all of the maps we have discussed are $G$-equivariant:
\[
  \begin{tikzpicture}[scale=2.3,small,>=stealth',node distance=.7em,baseline=(current bounding box.north)]
      \node (inertial) at (2.2,.3) {$\GKt$};
      \node (spacetime) at (0,-.5) {$\GH$};
      \node (hyper) at (.7,-.5) {$\Hyp$};
      \node (conf) at (2.2,-1.2) {$\GP$};
      \node (pastcone) at (1.4,-.5) {$\pastcone$};
      \node (obs) at (1,1) {$\GK$};
      \node [below of=obs,font=\fontsize{7}{7},color=blue] {\sf observer space}
          edge [->]  node [near end] (base-event) {} (spacetime)
          edge [->] node [right,pos=.6] (asym) {$\past$} (conf)
          edge [->]  node (extend) {} (inertial)
          edge [->]  node (base-now) {}(hyper)
          edge [->]  node (pastconemap) {} (pastcone);
      \node [below of=spacetime,font=\fontsize{7}{7},color=blue] {\sf spacetime};
      \node [below of=hyper,node distance=1.1em,vsmall,color=blue,text centered,text width=1.5cm] {\sf hyperbolic space};
      \node [below of=conf,vsmall,color=blue] {\sf conformal 3-sphere};
      \node [below of=inertial,node distance=1.1em,vsmall,color=blue,text centered,text width=1.8cm] {\sf inertial observer space}
           edge [->] (conf);
      \node [below of=pastcone,node distance=1.1em,vsmall,color=blue,text centered,text width=1.8cm] {\sf ambient past light cone}
          edge [->] node [near end] (canon) {} (conf);          
  \end{tikzpicture}
\qquad
  \begin{tikzpicture}[small,scale=2,>=stealth',baseline=(current bounding box.north)]
      \node (conf) at (2.4,-1.2) {$[x-u]$};
      \node (spacetime) at (0,-.5) {$x$};
      \node (hyper) at (.7,-.5) {$u$};
      \node (pastcone) at (1.3,-.5) {$x-u$}
          edge [|->,node distance=1em] (conf);
      \node (inertial) at (2.4,.2) {$\gamma(\R)$}
          edge [|->] (conf);
      \node [circle,inner sep=.4em] (obs) at (1,1) {$(x,u)$}
          edge [|->]  (spacetime)
          edge [|->]  (conf)
          edge [|->]   (inertial)
          edge [|->]  (hyper)
          edge [|->]  (pastcone);
 \end{tikzpicture}
\]
where $\gamma$ is the geodesic (\ref{geodesic}) through the observer. 

Since each of the spaces is homogeneous, they can all be described as \define{Klein geometries}---quotients $G/G'$ where $G'$ is the closed subgroup of $G$ stabilizing an element of the space (see e.g.\ \cite{Sharpe}).  Any equivariant map $G/G'' \to G/G'$ is then induced by some inclusion $G'' \to G'$ of one subgroup of $G$ into another.  Fixing an arbitrary basepoint $(x,u)\in \GK$, the above maps give basepoints in each of the other spaces.  The corresponding stabilizer subgroups and inclusions are then:
\begin{equation}
\begin{array}{l@{\;\cong\;}ll}
K & \SO(3) &   \text{\sf stabilizer of  $(x,u)\in \GK$} \\
\Kt & \SO(3)\times \R & \text{\sf stabilizer of $\gamma(\R)\subset \GH$, with $\gamma$ as in (\ref{geodesic})} \\
H & \SO_o(3,1) & \text{\sf stabilizer of } x 
                               \in \GH \\
H' & \SO(4) &   \text{\sf stabilizer of } u
                               \in \Hyp \\
\Hpp & \ISO(3) & \text{\sf stabilizer of } x-u \in \pastcone \\
P & \SIM(3) &    \text{\sf stabilizer of } [x-u]\in \GP \\
\end{array}
\qquad
  \begin{tikzpicture}[scale=1,vsmall,>=stealth',fill=black!10,baseline=(current bounding box.center)]
      \node (conf) at (2.2,-1.2) {$P$};
      \node (spacetime) at (0,-.5) {$H$};
      \node (hyper) at (.7,-.5) {$H'$};
      \node (pastcone) at (1.3,-.5) {$\Hpp$}
          edge [->] (conf);
      \node (inertial) at (2.2,.3) {$\Kt$}
          edge [->] (conf);
      \node (obs) at (1,1) {$K$}
          edge [->]  (spacetime)
          edge [->]  (conf)
          edge [->]   (inertial)
          edge [->]  (hyper)
          edge [->]  (pastcone);
      \begin{pgfonlayer}{background} 
           \filldraw [fill=black!4,draw=black!2]
              (conf.south -| spacetime.west) rectangle (obs.north -| conf.east);
      \end{pgfonlayer}
  \end{tikzpicture}
\end{equation}
It is convenient to describe these subgroups explicitly in a matrix representation, since all of our spaces are subspaces, or else quotients of subspaces, of $\R^{4,1}$.

Defining $\R^{4,1}:=(\R^5,\eta)$ with $\eta$ the matrix of the standard Minkowski metric on $\R^5$, $G\subset \GL(5)$ is then the connected component of the group of matrices $A$ such that $A^T\eta A = \eta$.  Its Lie algebra $\g\subset \mathfrak{gl}(5)$ consists of all matrices $A$ such that $A^T\eta = -\eta A^T$. As the basepoint for observer space $\GK$, we choose $(x,u)$ with 
\begin{equation}
\label{basepoint}
 x = 
{\small   \left(
  \begin{array}{c}
0\\ 0 \\ 1
  \end{array} 
  \right)}
 , \quad 
 u =
{\small   \left(
  \begin{array}{r}
1\\ 0 \\ 0
  \end{array} 
  \right)}
\end{equation}
where we adopt the convention for such column vectors that the middle entry is really a three-component vector.
Then the subgroups $H, H', K\subset G$ are easily described: $H$ is the upper $4\times 4$ block of $G$, $H'$ the lower $4\times 4$ block, and $K$ is their intersection, the $3\times 3$ block in the middle.  Likewise for the Lie algebras $\h$, $\h'$ and $\k$. 

For the groups related to conformal geometry, it is convenient to temporarily change basis so that the asymptotic past and future of our base observer, corresponding to $[x-u]$ and $[x+u]$ in $\R^{4,1}$, are two of the coordinate axes.  We define $\overline{\R^{4,1}}:=(\R^5,\overline\eta)$ where 
\[
\overline\eta :=  
  \left(
  \begin{array}{ccc}
  0&0&1\\
  0&1_3&0\\
  1&0&0\\
  \end{array}
  \right)
 =S\eta S^T \,, \qquad 
  S=
  \left(
  \begin{array}{ccc}
  \frac{1}{\sqrt 2}&0&\frac{1}{\sqrt 2}\\
  0&1_3&0\\
  -\frac{1}{\sqrt 2}&0&\frac{1}{\sqrt 2}\\
  \end{array}
  \right) 
\]  
so that $S$ is the matrix of a linear isometry $S\maps  \R^{4,1} \to \overline{\R^{4,1}}$.  For any subgroup $G'\subseteq G$ with Lie alebra $\g$, we have $\overline G := SGS^{-1}\subset \GL(5)$ and $\overline \g := S\g S^{-1}\subset \mathfrak{gl}(5)$ as the corresponding symmetry group and Lie algebra on $\overline{\R^{4,1}}$.  In particular, conjugating an arbitrary matrix in $\g$, we find that an element of $\overline\g$ has the form
\begin{equation}
\label{gbar-element}
  \left(
  \begin{array}{ccc}
  \tau&-q^T&0\\
 p &b&q\\
  0&-p^T&-\tau\\
  \end{array}
  \right) 
\end{equation}
where the $3\times 3$ matrix $b$ is anti-symmetric.  The Lie algebra $\overline \p$ of the stabilizer $\overline P$ of the point $[S(x-u)]$ is then the set of such matrices for which $(0,0,1)^T$ is an eigenvector, namely those for which $q=0$.  Likewise, $\overline\go$, the stabilizer algebra of the intertial observer, is the intersection of $\p$ with the stabilizer of $[S(x+u)]$, and hence consists of matrices in $\overline \g$ for which $p=q=0$.  

By exponentiating an element (\ref{gbar-element}) with $q=0$, we find an explicit form for elements of $\overline P$, and check that they can be uniquely factored into the form 
\begin{equation}
\label{factorize}
  \left(
  \begin{array}{ccc}
  e^{\tau}& 0 &0\\
  0&B&0\\
  0&0&e^{-\tau}\\
  \end{array}
  \right) 
  \left(
  \begin{array}{ccc}
  1&0 &0\\
  p&1_3&0\\
  -\frac12 p^Tp&-p^T&1\\
  \end{array}
  \right) 
\end{equation}
Here $B\in \overline K=K\cong \SO(3)$, $\tau\in \R$.  The first factor is clearly the exponential of an element of $\go$, while the second is a multiplicative rewriting of the additive group $\R^3$.   The factorization makes $\overline P$ and hence also $P$, switching back to $\R^{4,1}$, a semidirect product:
\[
    P \cong \Kt\ltimes \R^3\cong (\SO(3)\times \R)\ltimes \R^3
\] 
which is just the group of Weyl or similarity transformations, $\SIM(3)$.   The group $\Hpp$ is clearly the subgroup of $P$ with $\tau = 0$, which is just $\ISO(3)\cong \SO(3) \ltimes \R^3$.

\section{Observer space from conformal geometry}

Each observer in de Sitter spacetime is part of exactly one of the local fields of observers (\ref{observer-field}), and the space of these observer fields is a conformal 3-sphere.  We now consider what additional structure on the conformal 3-sphere specifies a particular observer within the corresponding field of observers.  Ultimately, this will let us produce `observer space geometries' from more general conformal geometries than just $\GP$.

To construct the de Sitter observer space from the conformal sphere, let us begin by constructing the {\em inertial} observer space.   The key is that an inertial observer is completely determined if we specify both its asymptotic past and future, and these may be chosen to be any distinct points of $\GP$. 

This is easy to prove using the ambient space $\R^{4,1}$.   Pick two points $[v]\neq [w]\in\GP$. Any observer $(x,u)$ with asymptotic past $[x-u]=[v]$ and asymptotic future $[x+u]=[w]$ must lie in the intersection of $[v,w]$ with $\GH$.  This intersection is a pair of timelike geodesics, but only one of these geodesics has $[v]$ as its past and $[w]$ as its future, while the other geodesic has it the other way around.  Thus $(x,u)$ is determined up to translation along the worldline.  On the other hand, we can always arrange for the representatives $v$ and $w$ to  $[v]$ and $[w]$ to be such that $v\in \pastcone$ and $w\in\futurecone$ with $\eta(v,w) = 2$.  This uniquely determines a point $(x,u)\in \GK\subset \GH\times \Hyp$ with 
\begin{align}
x-u &= v \nonumber\\ x+u &=w  \label{linsys}
\end{align}
This whole process is $G$-equivariant, and thus we have shown:

\begin{prop}
\label{prop:inertial}
The inertial observer space $\GKt$ of de Sitter spacetime is isomorphic as a $G$-space to the space of ordered pairs of distinct points in $\GP$. 
\end{prop}

Picking an ordered pair of points in $\GP$, and calling them `$0$' and `$\infty$', is the same as giving $\GP-\{\infty\}$ the structure of a conformal vector space, or a Euclidean vector space up to scale.  From  the previous section we know the stabilizer of `$\infty$' is isomorphic to $\SIM(3)$, the group of transformations of a scale-free Euclidean affine space.  Specifying the origin `$0$' then makes this affine space a scale-free Euclidean vector space, reducing the $\SIM(3)$ symmetry to $\Go= K\times \R$.  The orbits of the $K\cong \SO(3)$ part are 2-spheres centered at the origin, while the the $\R$ part acts as dilations, moving between these concentric spheres:  
\begin{center}
\begin{tikzpicture}
    \node[anchor=south west,inner sep=0] (image) at (0,0) {\includegraphics[width=4cm]{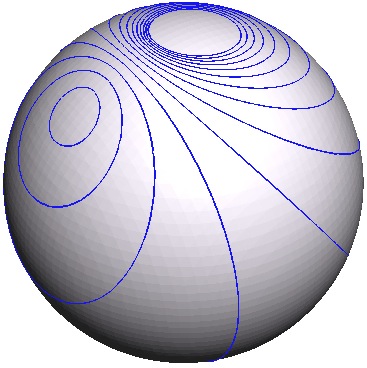}};
    \begin{scope}[x={(image.south east)},y={(image.north west)}]
    \node [red,inner sep=.3mm] (infinity) at  (0.53,0.92) {$\bullet$};
    \node at (.6,1.1) {`$\infty$'} edge [black,very thin]  (infinity);
    \node [vsmall] at (.95,1.05) {\sf (asymptotic future)};
    \node [red,inner sep=.3mm] (origin) at  (.2,.69) {$\bullet$};
    \node at (-.2,.8) {`$0$'}edge [black,very thin]  (origin);
    \node [vsmall] at (-.25,.7) {\sf (asymptotic past)};
    \node at (-.3,.25) {$S^3$} edge [black,very thin] (.15,.35);
    \node [small,text width=2cm] at (1.5,.4) {\sf equally spaced concentric copies of $S^2$}
         edge [black,very thin] (.95,.57)
         edge [black,very thin] (.82,.45)
         edge [black,very thin] (.67,.14);
    \end{scope}
\end{tikzpicture}
\end{center}
This is just a picture of an {\em inertial observer} from the holographic perspective.  

From here, it is easy to get to the observer space.  From the spacetime perspective, the flow of time for an inertial observer is given by the action of the $\R$ part of $\Go$.  Since we have an isomorphism as $G$-spaces, this must be the case in the holographic picture as well: {\em an inertial observer's timeline is the action of $\R$ as dilations on the corresponding conformal vector space}.  Using Klein geometry, we need only pick some feature in inertial observer space that is stabilized by $K$, so that the homogeneous space of such features becomes $G/K$.  The obvious choice is one of the copies of $S^2$ determined by an inertial observer.   We immediately get:
\begin{prop}
The observer space $\GK$ of de Sitter spacetime is isomorphic as a $G$-space to the space of ordered pairs of distinct points in $\GP$ together with a fixed sphere centered at the origin of the resulting conformal vector space. 
\end{prop}

We sketched in the introduction how this corresponds to the spacetime picture: the region inside the chosen sphere represents the portion of the infinite past that is causally connected with the observer at the observer's current time, so that the expansion of this sphere is the measure of time along the observer's worldline.  We can now see more precisely how this works.  In de Sitter spacetime, the light cone at $x\in\GH$ is $(x+[x]^\perp)\cap \GH$.  The affine plane $x+[x]^\perp$ intersects every light ray in the ambient past light cone $\pastcone$ on one side of $[x]^\perp$.  That is, the asymptotic past of the light cone is the sphere in $\GP$  consisting of all light rays in the subspace $[x]^\perp$: 
\[
    \{ [v] : v\in [x]^\perp, \eta(v,v)=0\}\cong S^2
\] 
which we declare to be the \define{unit sphere}.  

Notice that the unit sphere does not depend on the observer's velocity in spacetime, but only on the event.  The unit sphere by itself does not determine the event, but it does so once we also specify a co-orientation, i.e.\ a notion of `inside' versus `outside' of the sphere.  We thus have:
\begin{prop}
De Sitter spacetime $\GH$ is isomorphic as a $G$-space to the space of co-oriented 2-spheres in $\GP$.
\end{prop}

To summarize the correspondence between spacetime and holographic pictures, we can draw all of the equivariant maps, but describing all of the spaces alternatively in relation to spacetime or in relation to conformal space.  These are simply two geometric representations of the same `observer space' Klein geometry $G/K$:
\[
  \begin{tikzpicture}[scale=2.3,small,>=stealth',node distance=.7em,baseline=(current bounding box.north)]
      \node (title) at (0,1.2) {\sf \large Spacetime picture:};
      \node (inertial) at (2.2,.3) {$\GKt$};
      \node (spacetime) at (0,-.5) {$\GH$};
      \node (hyper) at (.7,-.5) {$\Hyp$};
      \node (conf) at (2.2,-1.2) {$\GP$};
      \node (pastcone) at (1.4,-.5) {$\pastcone$};
      \node (obs) at (1,1) {$\GK$};
      \node [below of=obs,node distance=1.1em,vsmall,color=blue,text centered, text width=2.2cm] {\sf space of unit future timelike vectors}
          edge [->]  node [near end] (base-event) {} (spacetime)
          edge [->] node [right,pos=.6] (asym) {$\past$} (conf)
          edge [->]  node (extend) {} (inertial)
          edge [->]  node (base-now) {}(hyper)
          edge [->]  node (pastconemap) {} (pastcone);
      \node (stlabel) [below of=spacetime,vsmall,color=blue,text centered,text width=1.4cm] {\sf spacetime};
      \node [below of=hyper,node distance=1.1em,vsmall,color=blue,text centered,text width=1.5cm] {\sf space of `nows'};
      \node [below of=conf,vsmall,color=blue] {\sf past boundary};
      \node [below of=inertial,node distance=1.1em,vsmall,color=blue,text centered,text width=2cm] {\sf space of timelike geodesics}
           edge [->] (conf);
      \node [below of=pastcone,node distance=.2em,vsmall,color=blue,text centered,text width=1.8cm] {} 
          edge [->] node [near end] (canon) {} (conf);          
      \begin{pgfonlayer}{background} 
           \filldraw [fill=red!10,draw=black!5]
              (stlabel.south -| stlabel.west) rectangle (spacetime.north -| stlabel.east);
      \end{pgfonlayer}        
  \end{tikzpicture}
\qquad
  \begin{tikzpicture}[scale=2.3,small,>=stealth',node distance=.7em,baseline=(current bounding box.north)]
      \node (title) at (0,1.2) {\sf \large Holographic picture:};
      \node (inertial) at (2.2,.3) {$\GKt$};
      \node (spacetime) at (0,-.5) {$\GH$};
      \node (hyper) at (.7,-.5) {$\Hyp$};
      \node (conf) at (2.2,-1.2) {$\GP$};
      \node (pastcone) at (1.4,-.5) {$\pastcone$};
      \node (obs) at (1,1) {$\GK$};
      \node [below of=obs,node distance=1.2em,vsmall,color=blue,text centered, text width=2.2cm] {\sf space of Euclidean \mbox{decompactifications}}
          edge [->]  node [near end] (base-event) {} (spacetime)
          edge [->] node [right,pos=.6] (asym) {$\past$} (conf)
          edge [->]  node (extend) {} (inertial)
          edge [->]  node (base-now) {}(hyper)
          edge [->]  node (pastconemap) {} (pastcone);
      \node [below of=spacetime,node distance=1.1em,vsmall,color=blue,text centered,text width=1.8cm] {\sf space of co-oriented spheres};
      \node [below of=hyper,node distance=1.1em,vsmall,color=blue,text centered,text width=1.5cm]  {}; 
      \node (conflabel) [below of=conf,vsmall,color=blue] {\sf conformal 3-sphere};
      \node [below of=inertial,node distance=1.1em,vsmall,color=blue,text centered,text width=1.8cm] {\sf ordered pairs of distinct points}
           edge [->] (conf);
      \node [below of=pastcone,node distance=1.1em,vsmall,color=blue,text centered,text width=1.8cm] {\sf tautological bundle}
          edge [->] node [near end] (canon) {} (conf);          
      \begin{pgfonlayer}{background} 
           \filldraw [fill=red!10,draw=black!5]
              (conflabel.south -| conflabel.west) rectangle (conf.north -| conflabel.east);
      \end{pgfonlayer}        
  \end{tikzpicture}
\]

One disadvantage of our construction so far of the observer space $\GK$ from conformal space is that it is rather `nonlocal,' involving a choice of  two points as well as a 2-sphere between them to set the scale.  This is quite natural for the conformal 3-sphere, but does not readily generalize to arbitrary conformal manifolds, as we will do shortly, using differential methods.  
Fortunately, there is an equivalent way to view an observer, involving only differential data at a single point of the conformal 3-sphere. 

A \define{transverse 3-plane} in $\pastcone$ is a 3-dimensional subspace of some tangent space $T_v\pastcone$ transverse to the fiber direction of the bundle (\ref{pastlightcone}).   Let us show that a transverse 3-plane in $\GP$ is equivalent to an observer in $\GH$.  Note that $T_v\pastcone\cong [v]^\perp$ so we can think of a transverse 3-plane instead as a 3d subspace $W\subset [v]^\perp\subset \R^{4,1}$ such that $[v]^\perp=W\oplus [v]$, or equivalently a spacelike subspace of $[v]^\perp$.  Then $W^\perp\subset \R^{4,1}$ is a 2-dimensional Lorentzian subspace whose projective cone consists of $[v]$ and one other lightlike subspace $[w]$.  The ordered pair $([v],[w])$ in $\GP$ gives an inertial observer via Prop.~\ref{prop:inertial}.  Since we also have the specific basepoint $v\in \pastcone$, we may assume $[w]$ is represented by $w\in\futurecone$ with $\eta(v,w) = 2$, as we have done once before, and solve (\ref{linsys}) for an observer $(x,u)\in\GK$.  Conversely, an observer $(x,u)\in \GK$ determines null vectors $v,w$ by (\ref{linsys}), and $[v,w]^\perp=[x,u]^\perp$ is a spacelike subspace of $[v]^\perp\cong T_v\pastcone$.  These processes are inverse and everything is $G$-equivariant.  Recalling that the bundle (\ref{pastlightcone}) is isomorphic to the tautological bundle of the conformal 3-sphere,  we can neatly summarize these observations.

\begin{prop}
The observer space of de Sitter spacetime is isomorphic as a $G$-space to the space of all transverse 3-planes in the tautological bundle over the conformal 3-sphere. 
\end{prop}

It is worth noting that a transverse 3-plane is the same as the derivative of a section of the bundle at a point, and that a section of the tautological bundle is just a metric in the conformal class.  So, this proposition lets us think of an observer as a `metric up to first order' at a point in conformal space, or in other words the metric at the point, together with the first derivative of the metic.

This description of the observer space $\GK\cong G/K$ is not only local, but is also easily applied to any conformal manifold, since it makes no use of special features of the conformally flat sphere:
\begin{defn}
\label{observer.space.of.conformal.space}
Given an arbitrary 3-manifold with conformal metric, we define its associated \define{observer space} to be the space of transverse 3-planes in its tautological bundle.   
\end{defn}   
This definition is the `holographic' analog of the definition of observer space associated to a spacetime as the unit future tangent bundle; it will be more fully justified in the next section, where we show that it is an `observer space geometry' in the sense defined in \cite{lifting}.   However, before moving on to this goal,  first we make a couple of corollary observations regarding inertial observers.

Recall that an {\em inertial} observer in $\GP$ is determined simply by an ordered pair of points $[v],[w]$ in $\GP$, and these give a subspace $W=[v,w]^\perp$.  Without a specific representative $v\in \pastcone$ of $[v]$, we cannot naturally view $[v]^\perp$ as a particular tangent space to the past light cone.  However, we can view it as simultaneously representing {\em all} tangent spaces to points of $\pastcone$ in the fiber over $[v]$.  In this way, $W$ determines a 3-dimensional distribution on the fiber, and this distribution is invariant under translation along the fiber.  

\begin{prop}
The inertial observer space of $\GH$ is isomorphic as a $G$-space to the space of $\R$-invariant distributions of transverse 3-planes supported on a single fiber of the tautological bundle.  
\end{prop}
If we smoothly pick such an invariant distribution on each fiber of (\ref{pastlightcone}), we have precisely the horizontal subspaces of an Ehresmann connection on the principal $\R$ bundle.   Let us define a \define{field of inertial observers on $\GP$} to be a section of the bundle $\iO \to \GP$.  We then have:

\begin{prop}
The space all fields of inertial observers on $\GP$ is isomorphic to the space of connections on the principal $\R$ bundle $\pastcone \! \to \GP$.  
\end{prop}

\section{The observer space of a general M\"obius geometry}

Observer space geometries are introduced in \cite{lifting} as deformations of the observer space of de Sitter spacetime, or one of the other homogeneous spacetimes.  The idea is a straightforward application of {\em Cartan geometry} \cite{Sharpe}.  Namely, since any Klein geometry gives a type of Cartan geometry modeled on it, we simply use the Klein geometry of de Sitter observer space as a model.
\begin{defn}
An \define{observer space geometry} is a Cartan geometry modeled on $G/K$, where $G=\SO_o(4,1)$ is  the connected de Sitter group and $K\cong \SO(3)$ is the group of spatial rotations around a fixed observer in de Sitter spacetime. 
\end{defn}

Roughly, a Cartan geometry is a manifold with the same `infinitesimal' geometry as its model Klein geomety, but without the same rigid uniformity on the macroscopic level.  More precisely, a \define{Cartan geometry} on a manifold $M$, with model $G/G'$, consists of a principal $G'$ bundle $Q$ over $M$ equipped with a $\g$-valued 1-form $A$, the {\bf Cartan connection}, satisfying three properties:
\begin{C-list}  
\item \label{cartan:nondeg} For each $q\in Q$, $A_q\maps T_q Q\to \g$ is a linear isomorphism;
\item \label{cartan:equivariant} $(R_g)^\ast A = \Ad(g^{-1})\circ A \quad \forall g \in G'$;
\item \label{cartan:mc} $A$ restricts to the Maurer--Cartan form on vertical vectors. 
\end{C-list}

Besides observer space geometries, we need conformal Cartan geometry, which is more general than the a manifold with a conformal metric.  
\begin{defn}
A \define{M\"obius geometry} is a Cartan geometry modeled on $G/P$, where $G=\SO_o(4,1)$ and $P$ is a parabolic subgroup, stabilizing some point in the projective light cone $\GP$ in $\R^{4,1}$.  
\end{defn}
\begin{prop}
\label{conf-equiv-prob}
A Cartan geometry on $M$ modeled on $G/P$ induces a conformal structure on $M$.  Conversely, a conformal structure on $M$ is induced by a unique normal Cartan geometry on $M$ modeled on $G/P$.\end{prop}
This is proved in Sec.\  7.3 of Sharpe's book \cite{Sharpe}. The definition of `normal' can also be found there (Def. 7.2.7), though here we do not need it in full generality, thanks to a dimensional coincidence: a M\"obius geometry on a {\em three}-dimensional manifold is \define{normal} if and only if its curvature takes values in the $P$ module $\gp$.  

Our goal here is to obtain observer space geometries from M\"obius geometry, hence from any conformal 3-manifold.  But first, we discuss a few basic facts about observer space geometries.  

The de Sitter observer space $G/K$ is reductive, meaning that the quotient representation $\g/\k$ of $K$ can be embedded as a subrepresentation of $\g$, complementary to $\k$.  In fact, $\g$ splits into four irreducible representations of $K$ corresponding to four types of infinitesimal transformations of an observer:
\begin{equation}
\label{obs-splitting}
 \begin{tikzpicture}[scale=1,>=stealth',baseline=(current bounding box.center)]
     \node (eqn) at (0,0) {$\g =\k \oplus (\ghk \oplus \vec \ggh \oplus \ggh_o)$};
     \begin{scope}[x={(eqn.east)},y={(eqn.north)}]
       \node (rot) [circle,inner sep=.5mm] at (-.46,-.5) {};
       \node (boost) [circle,inner sep=.5mm] at (-.04,-.5) {};
       \node (strans) [circle,inner sep=.5mm] at (.36,-.5) {};
       \node (ttrans) [circle,inner sep=.5mm] at (.7,-.5) {};
       \node [small,text width=1.8cm,text centered,anchor=north] at (-1,-3) {\sf rotations \mbox{(stabilizer)}} edge [very thin] (rot);
       \node [small,anchor=north] at (-.2,-3) {\sf boosts} edge [very thin] (boost);
       \node [small,text width=1.8cm,text centered,anchor=north] at (.6,-3) {\sf spatial \mbox{translations}} edge [very thin] (strans);
       \node [small,text width=1.8cm,text centered,anchor=north] at (1.7,-3) {\sf time \mbox{translations}} edge [very thin] (ttrans);
     \end{scope}
 \end{tikzpicture}
\end{equation}
where the parenthesized terms, a $K$ representation isomorphic to $\g/\k$, make up the tangent space to de Sitter observer space $G/K$ at the basepoint stabilized by $K$.  The reductive splitting allows the Cartan connection to be split into four separate fields on observer space with distinct physical roles \cite{lifting}.  

On the other hand, the present context suggests a different decomposition of $\g$---one very often used in conformal geometry (see e.g.\ \cite{leitner}), but here reinterpreted geometrically in terms of observer space:
\begin{equation}
\label{conf-splitting}
 \begin{tikzpicture}[scale=1,>=stealth',baseline=(current bounding box.center)]
     \node (eqn) at (0,0) {$\g = \gm \oplus \go \oplus \gp$};
     \begin{scope}[x={(eqn.east)},y={(eqn.north)}]
       \node (go) [circle,inner sep=.5mm] at (.22,-.7) {};
       \node (gm) [circle,inner sep=.5mm] at (-.4,-.7) {};
       \node (gp) [circle,inner sep=.5mm] at (.7,-.7) {};
       \node [small,text width=1.8cm,text centered,anchor=north] at (.22,-3) {\sf rotations and time \mbox{translations}} edge [very thin] (go);
       \node [small,text width=2.5cm,text centered,anchor=north] at (-1.28,-3) {\sf \mbox{translations of} \mbox{asymptotic past}} edge [very thin] (gm);
       \node [small,text width=2.5cm,text centered,anchor=north] at (1.72,-3) {\sf \mbox{translations of} \mbox{asymptotic future}} edge [very thin] (gp);
     \end{scope}
 \end{tikzpicture}
\end{equation}
This is a $\Z$-grading of $\g$ (where grade $k$ is trivial for all $|k|>1$), and explains why we have called the stabilizer of an inertial observer $\Go$ all along.  We have defined $P$ to be the stabilizer of the asymptotic past of an observer, and its Lie algebra is $\p=\go\oplus \gp$.  The Lie algebras of $H$, $H'$, $\Go$  and $\Hpp$ are $\h = \k\oplus \ghk$, $\h'= \k \oplus \vec \ggh$, $\hpp= \k\oplus \gp$, and $\go=\k \oplus \ggh_o$.  To complete the relationship between (\ref{obs-splitting}) and (\ref{conf-splitting}) we note there is a canonical isomorphism of $K$ representations $f\maps \ghk \to \vec\ggh$, and this lets us define the $K$ representations $\gm$ and $\gp$ (in fact invariant under the larger group $\Go$) by
\[
 \gpm = \{ (\xi,\zeta) \in \ghk \oplus \vec\ggh : \zeta = \pm f(\xi)\}
\]
Geometrically, this says, for example, that in order to translate the asymptotic future while fixing the asymptotic past, one can translate and boost in the same spatial direction, with the same magnitude according to our normalization.

We now prove our main mathematical result, justifying our preliminary definition (Def.~\ref{observer.space.of.conformal.space}) of the observer space of a conformal manifold. 
\begin{thm}
\label{mainthm}
A M\"obius geometry canonically induces an observer space geometry.  Moreover, this geometry may be identified as the space of transverse 3-planes in the tautological $\R$ bundle corresponding to the conformal metric induced by the M\"obius geometry. 
\end{thm}
\begin{proof}
Consider a Cartan geometry modeled on $G/P$, with $\pi\maps\ff\to \S$ the principal $P$ bundle, $A$ the Cartan connection.  Then $\fo:=\ff/K$ is an observer space geometry: the map $\ff \to \fo$ is a principal $K$ bundle and properties \ref{cartan:nondeg}, \ref{cartan:equivariant}, and \ref{cartan:mc} for the stabilizer group $P$ clearly imply the same properties for the subgroup $K\subset P$.  Thus, we need only show that $\fo$ may be identified with the space of transverse 3-planes in the tautological $\R$ bundle over $\S$. 

We first construct a map $\pi'\maps \ff \to \Met$ making $\ff$ a principal $\Hpp$ bundle over the tautological $\R$ bundle $\Met$ over $\S$.  Fix an inner product $\delta$ on $\g/\p$, invariant under $\Ad(P)$ up to a scale, i.e.\ an element of the tautological $\R$ bundle over the model space at the basepoint.  A Cartan connection $A$ gives in particular, for each $f\in \ff$, a coframe at $\pi(f)$: a linear isomorphism 
\begin{equation}
\label{coframe}
 e_{\pi(f)}\maps T_{\pi(f)}\S \to \g/\p.  
\end{equation}
Given $f\in \ff$, we then define $\pi''(f)\in \Met$, an inner product at the point $\pi(f)$, to be the pullback of the fixed inner product $\delta$.  We thus get a factorization of $\pi$:
\[
  \begin{tikzpicture}[>=stealth]
     \node (base) at (0,0) {$\S$};
     \node  (met) at (1,1) {$\Met$}
        edge [->] (base);
    \node (total) at (0,2) {$\ff$}
        edge [->] node [left] {$\pi$} (base)
        edge [->] node [above right] {$\pi'$} (met);
  \end{tikzpicture}
\]  
$A$ becomes a Cartan connection on the bundle $\pi'$, giving $\Met$ the structure of a Cartan geometry modeled on $G/\Hpp$.

Next, we construct a map $\pi''\maps \ff \to \fo$ making $\ff$ a principal $K$ bundle over $\fo$, where $\fo$ is the space of all transverse 3-planes to $\Met$.  The Cartan connection on $\Met$ gives in particular for each $f\in \ff$ a coframe at $\pi'(f)$, which we may think of as
\begin{equation}
\label{coframe2}
 e_{\pi'(f)}\maps T_{\pi'(f)}\Met \to (\gm\oplus \ggh_o).  
\end{equation}
The identification of $\g/\h''$ with $\gm\oplus \ggh_o$  is $K$-invariant, and we get a principal $K$ bundle $\pi''\maps \ff \to \fo$ by defining
\[
   \pi''(f) =  e_{\pi'(f)}^{-1}(\gm).
\]
This gives a factorization of $\pi'$:
\[
  \begin{tikzpicture}[>=stealth]
     \node (base) at (0,0) {$\S$};
     \node  (met) at (.6,1) {$\Met$}
        edge [->] (base);
    \node (3planes) at (1.2,2) {$\fo$}
        edge [->] (met);   
    \node (total) at (0,3) {$\ff$}
        edge [->] node [left] {$\pi$} (base)
        edge [->] node [right] {$\pi'$} (met)
        edge [->] node [above right] {$\pi''$} (3planes);
  \end{tikzpicture}
\]  
and $A$ is the Cartan connection on $\pi''\maps \ff\to \fo$.   Since $\ff$ is a principal bundle $K$ bundle over $\fo$, we have $\fo\cong \ff/K$, so the geometry we have constructed here is isomorphic to the induced geometry on $\ff/K$ described in the first paragraph of this proof.  
\end{proof}

Combining this with Prop.~\ref{conf-equiv-prob} we get a canonical observer space geometry from a conformal metric, via the canonical normal M\"obius geometry.  Of course, the canonical geometry may not be the physically relevant one, which might also have curvature components in $\gm\oplus \go$, just as there may be physical reasons for using a geometry with torsion (curvature components in $\vec\ggh \oplus \ggh_o$ for a spacetime Cartan connection) rather than the canonical torsion-free geometry.   

It is worth emphasizing that what we are doing is quite different from starting with 3-dimensional conformal manifold and constructing a 4-dimensional spacetime with that conformal manifold as part of its boundary \cite{FG}.   Rather, from a 3-dimensional conformal manifold, we construct a 7-dimensional manifold of observers.  Spacetime is a would-be quotient of this observer space, obtained by collapsing the `boost' directions---though in fact these boost directions need not be integrable, in which case spacetime does not exist as a natural quotient.  We explain much more about these ideas in \cite{lifting}, where we also give a sufficient criterion for integrability (though not a necessary criterion, as is crucial in \cite{hohmann}), based on certain curvature components vanishing.   Perhaps most interesting, however, is the case where there is a slight failure of integrability, and spacetime becomes an observer-dependent notion, though in the present context with the `conformal boundary' of spacetime remaining perfectly coherent.

\section{Discussion}

What we have called `holographic special relativity' is about encoding observers in the de Sitter universe using three-dimensional conformal geometry rather than spacetime geometry.  While the two perspectives are isomorphic, they suggest distinct types of Cartan-geometric deformation, making it tempting to propose a `holographic' analog of {\em general} relativity.  This would mean introducing a dynamical theory of conformal geometry, with the associated observer space given by Thm.\ \ref{mainthm}, but where spacetime, or at least spacetime in the usual Lorentzian sense, may become only an approximate concept.  This is just the opposite of what happens in ordinary general relativity, where spacetime remains coherent, but conformal space is no longer a quotient of the associated observer space as it is in the de Sitter universe.  

What sort of dynamical conformal geometry might lead to a physically realistic `holographic general relativity'?  One possibility is that it is to be found already in a recently proposed alternative to general relativity called `shape dynamics.'   This is also a theory in which (3+1)-dimensional spacetime geometry appears to be exchanged for 3-dimensional conformal geometry.  Remarkably, it is equivalent, at the level of classical field theory, to the ADM formulation of general relativity under certain conditions.  Both theories are particular gauge fixings of a `linking theory', and the constraint surface of each can be viewed as a particular gauge-fixing surface of the other, so that on the intersection of these surfaces the classical field theories are indistinguishable \cite{shape}.

On the other hand, the complete geometric description of a gravity theory must involve more than an analysis of classical field theories.   After all, gravity is not a field theory living on some space with a predefined geometry, but a theory that determines the physical geometry.  What is the physical meaning of the conformal space used in shape dyanamics, and how does this conformal space relate to the physical space we see around us, or to some other sort of geometric structure in the usual spacetime picture?  If shape dynamics is to take the place of general relativity, then what theory takes the place of {\em special} relativity?  Is it holographic special relativity as described here, or something else?  These are all questions whose answers should ultimately not involve phase space, canonical analysis, gauge fixing, or other techniques of classical field theory.  Cartan geometry gives a precise framework within which to study such questions.  

Some work toward understanding shape dynamics in such terms has recently been done one dimension down, where spacetime has 2+1 dimensions \cite{gryb}.    These authors essentially use the relationship between the Lie algebra splittings (\ref{obs-splitting}) and (\ref{conf-splitting}) (or rather their analogs for $\g=\so(3,1)$) to rewrite spacetime fields as fields on the conformal 2-sphere.  They find results consistent with our conjecture: the procedure turns 3-dimensional Chern--Simons gravity into shape dynamics.  However, while this work was motivated by Cartan geometry (in part through discussions with the present author), the Cartan geometric picture is not made explicit.  Moreover, as usual, while (2+1)-dimensional gravity is more tractable, it is also special enough that it could be misleading.  Only a careful geometric study in  3+1 dimensions will tell.

If these ideas are right, they have strong implications for the interpretation of shape dynamics.  For example, shape dynamics is so far thought of as a theory with `spatial conformal symmetry' rather than the refoliation symmetry of spacetime-based Hamiltonian approaches, like ADM.  However, if shape dynamics is really a `holographic general relativity' of the kind we have suggested, then the adjective `spatial' may not be quite appropriate: the relevant conformal space may not represent `space' as we usually understand it.  Indeed, in holographic special relativity, the conformal 3-sphere $\GP$ represents not `space' but a certain space of extended families of observers.  Points of `space', in the usual sense of the 3-dimensional world we see around us, are {\em nonlocal} features of conformal space.  Indeed, one might suspect this to be the geometric reason behind the nonlocal nature of the shape dynamics Hamiltonian.

Whether or not holographic general relativity turns out to take shape in the form of `shape dynamics', it is interesting to ponder the implications of alternative ways to represent observers.   Minkowski insisted that  only ``a kind of union'' (``eine Art Union") of space and time will endure.  The question we have been asking here is: {\em What kind of union?}  We have seen that spacetime is not the only possible answer.  To emphasize this point, it is worth pointing out some other possible answers, besides the holographic special relativity we have described.   Each different possibility comes with different potential consequences, both for generalizing to dynamical geometry and for quantization.

For example, another alternative to spacetime is implicit in what we have already explained.  We defined the observer space $\GK$  to be the unit future tangent bundle of de Sitter spacetime $\GH$ and found it convenient to do this using the embedding of $\GH$ and $\Hyp$ into the ambient space $\R^{4,1}$:
\[
   \GK = \{(x,u) \in \GH\times \Hyp : \eta(x,u) = 0\}.
\]
Notice, however, that this definition is completely symmetric.  We could just as well have defined $\GK$ to be the unit tangent bundle of $\Hyp$, and indeed the unit future tangent bundle of $\GH$ and the unit tangent bundle of $\Hyp$ are $G$-equivariantly isomorphic.  The hyperbolic space $\Hyp$ is the space of all possible `nows'---all notions of simultaneity that observers in de Sitter spacetime can have.   Thus, de Sitter observers are uniquely specified by the direction in which they are passing through this space of nows.  Because of the temporal peculiarities of de Sitter spacetime, two observers with the same `now' can even move in opposite directions in this space, though these two are on opposite sides of a cosmological horizon, hence cannot interact.  

Also, one can derive a version of `holographic special relativity' using anti de Sitter spacetime, rather than de Sitter spacetime.  The procedure is similar to what we have done here, if somewhat more involved.  The conformal boundary is at spacelike infinity, with Lorentzian signature, and observers' geodesics do not meet it.  Instead, an observer has a canonical notion of `space'---the orthogonal complement of the velocity vector.  This `spatial' subspace of the tangent space extends to a totally geodesic spacelike hypersurface in spacetime, and this leaves a 2-dimensional footprint at the conformal boundary.  

Another related way to represent observer space lies at the foundations of twistor theory, and what we have done here is admittedly similar in spirit to the twistor approach, if not in the details.   In twistor theory, a primary object of study is the 5-dimensional space $P(N)$ of all 1-dimensional null affine subspaces of $\R^{3,1}$.  Much like in our context, spacetime is viewed as a secondary construction, in this case given as the space of 2-spheres in $P(N)$, which can be identified with the celestial spheres at points of spacetime.  Different observers at the same point are related by conformal transformations of this celestial sphere \cite{PR}.

In Klein geometric language, this means the observer space of Minkowski spacetime is $\ISO_o(3,1)$-isomorphic to space of ways to realize spheres in $P(N)\cong\ISO_o(3,1)/(\SIM(2)\times \R)$ as standard spheres.  The space $P(N)$ is a $\Lambda\to 0$ limit of its obvious analog in de Sitter spacetime.  This space is  the boundary of the space $\GKt$ of inertial observers, which are just timelike geodesics.  

In any case, the present work serves as a further example of the seemingly universal nature of observer space geometry for theories of space and time.  

\subsection*{Acknowledgments}
I thank Julian Barbour for interesting discussions and for hospitality; my attempts to understand his ideas about conformal space replacing spacetime led me to write this article.   Conversations with James Dolan clarified the geometric picture.  I also thank Sean Gryb for discussions.

\end{document}